\documentclass[11pt]{article}

\newif\ifdraft \drafttrue
\newif\iffull \fulltrue

\draftfalse\fulltrue

\usepackage{amsmath,amsthm}
\usepackage{amsfonts}
\usepackage{amssymb, xcolor}

\usepackage{fullpage}

\usepackage{color}
\usepackage{times}
\definecolor{DarkGreen}{rgb}{0.1,0.5,0.1}
\definecolor{DarkRed}{rgb}{0.5,0.1,0.1}
\definecolor{DarkBlue}{rgb}{0.1,0.1,0.5}

\newcommand{\ar}[1]{\ifdraft \textcolor{red}{[AR: #1]}\fi}

\newcommand{\ie}{{\it i.e.}}
\newcommand{\BEAS}{\begin{eqnarray*}}
\newcommand{\EEAS}{\end{eqnarray*}}
\newcommand{\BEA}{\begin{eqnarray}}
\newcommand{\EEA}{\end{eqnarray}}
\newcommand{\BEQ}{\begin{equation}}
\newcommand{\EEQ}{\end{equation}}
\newcommand{\BIT}{\begin{itemize}}
\newcommand{\EIT}{\end{itemize}}
\newcommand{\BNUM}{\begin{enumerate}}
\newcommand{\ENUM}{\end{enumerate}}

\newcommand{\eps}{\epsilon}

\newcommand{\Mpna}{\mathcal{M}_{\mathcal{P}, n, \alpha, \beta, \epsilon}}
\newcommand{\pd}{\mathcal{P}}

\newcommand{\hbi}{\hat{b}_i}

\newcommand{\tp}{\tilde{p}}

\newcommand{\pbi}{p_{b_i}}
\newcommand{\phbi}{p_{\hbi}}
\newcommand{\pz}{p_0}
\newcommand{\po}{p_1}

\newcommand{\E}{\mathbb{E}}

\DeclareMathOperator*{\argmin}{arg\,min}

\newtheorem{assumption}{Assumption}[section]
\newtheorem{definition}{Definition}[section]

\newtheorem{lemma}{Lemma}[section]
\newtheorem{proposition}{Proposition}[section]
\newtheorem{theorem}{Theorem}[section]

\newtheorem{remark}{Remark}[section]

\begin{document}
\title{Buying Private Data without Verification}

\author{
Arpita Ghosh\thanks{Cornell University} \and
Katrina Ligett\thanks{Caltech. Supported in part by an NSF CAREER award (CNS-1254169), the US-Israel Binational Science Foundation (grant 2012348), the Charles Lee Powell Foundation, a Google Faculty Research Award, and a Microsoft Faculty Fellowship.} \and
Aaron Roth\thanks{University of Pennsylvania. Supported in part by an NSF CAREER award, under
NSF grants CCF-1101389 and CNS-1065060, and a Google Focused Research Award.} \and
Grant Schoenebeck\thanks{University of Michigan} 
}
\maketitle

\begin{abstract}
We consider the problem of designing a survey to aggregate non-verifiable information from a privacy-sensitive population: an analyst wants to compute some aggregate statistic from the private bits held by each member of a population,  but cannot verify the correctness of the bits reported by participants in his survey. Individuals in the population are strategic agents with a cost for privacy, \ie, they not only account for the payments they expect to receive from the mechanism, but also their privacy costs from any information revealed about them by the mechanism's outcome---the computed statistic as well as the payments---to determine their utilities.
How can the analyst design payments to obtain an accurate estimate of the population statistic when individuals strategically decide {\em both} whether to participate and whether to truthfully report their sensitive information?

We design a differentially private peer-prediction mechanism~\cite{MRZ05} that supports {\em accurate} estimation of the population statistic as a Bayes-Nash equilibrium in settings where agents have explicit preferences for privacy. The mechanism requires knowledge of the marginal prior distribution on bits $b_i$, but does not need full knowledge of the marginal distribution on the costs $c_i$, instead requiring only an approximate upper bound. Our mechanism guarantees $\epsilon$-differential privacy to each agent $i$ against any adversary who can observe the statistical estimate output by the mechanism, as well as the payments made to the $n-1$ other agents $j\neq i$. Finally, we show that with slightly more structured  assumptions on the privacy cost functions of each agent ~\cite{CCKMV13}, the cost of running the survey goes to $0$ as the number of agents diverges.
\end{abstract}

\section{Introduction}

Consider the problem faced by a researcher who would like to compute some (unknown) statistic about a target population: for example, the prevalence of AIDS among university professors. He could, of course, run a survey and `just ask'; the obstacle he faces is that an individual's AIDS status is sensitive data, and the individuals being surveyed may be concerned that some harm might befall them if they were to participate in such a survey. This is a problem for the researcher, since individuals are not obligated to participate in his survey---to solve this problem, he might offer a (possibly different) payment to each participant to compensate them for such concerns and offset their `privacy' costs from participating in his survey. The researcher still has a problem, though---even if he manages to recruit individuals to participate in his survey (for instance, incentivized by the payment offered for participation), these participants are still not obligated to truthfully report their private data, {\em nor} does our researcher have any direct means to verify the truth of their responses.

The ``sensitive surveyor's problem''~\cite{GR11} attempts to model the problem of procuring data from individuals, each with a cost of privacy, in order to estimate some aggregate statistic
of a population, and has since been studied in some depth \cite{GR11,FL12,RS12,LR12,GL13,NVX14}. This stream of work, however, has thus far only modeled settings in which the data collected from each participant are {\em verifiable}, and the only information that a participant can misreport is her privacy cost determining the compensation she requires for participation---that is, an individual may choose to not participate in the survey, but if participating, she cannot lie about her data.\footnote{For instance, her data will be automatically collected from some database once she consents to its use, or equivalently, false reports can easily be identified, in which case not compensating individuals for false reports effectively eliminates the incentive to participate and yet lie about the sensitive data.}
As evident from our cartoon AIDS survey example, however, there are clearly also scenarios where a researcher {\em cannot} directly verify the truth of the responses to his survey. A researcher in such a scenario faces an additional problem in computing a reliable estimate of an aggregate statistic beyond cheaply eliciting participation---recruiting a large enough population of respondents does not automatically lead to representative estimates because individuals may still lie about their data, compromising the accuracy of the estimate. How can payments be designed to ensure that the survey produces an accurate estimate of the population statistic, when individuals strategically decide both whether or not to participate {\em and} whether or not to truthfully report their sensitive data, accounting both for the payments they expect to receive from the mechanism and their privacy costs from the mechanism's outcome?

The problem of eliciting unverifiable\footnote{\ie, where agents' reports cannot be compared against some ground truth for verification} information from strategic agents is addressed by the peer prediction mechanism~\cite{MRZ05}, which uses proper scoring rules to design payments such that truthful reporting is a (Bayes)-Nash equilibrium of the mechanism.
Agents in~\cite{MRZ05}, however, have no concern for privacy---they do not care about what the outcomes computed by the mechanism might reveal about their reports, and act only to maximize their expected payment from the mechanism. The peer-prediction mechanism, therefore, does not directly address the following problem of eliciting information from privacy-sensitive individuals.

Consider an abstract model of such a privacy-sensitive population: each agent in the population has a private (sensitive) bit $b_i$, as well as a parameter $c_i$ which bounds her cost for privacy as follows. If the mechanism that agents interact with is $\epsilon$-differentially private~\cite{DMNS06}, then an agent's cost (which is potentially distinct for each outcome) is at most, though not necessarily exactly, $\epsilon c_i$. The bit-cost pairs $(b_i,c_i)$ for each individual are drawn from a (not necessarily independent) commonly known prior distribution. Implementing a peer prediction mechanism in such a setting faces two main obstacles:
\begin{enumerate}
\item The way that the peer prediction mechanism incentivizes truth-telling despite unverifiability is by rewarding a participant's stated bit's correlation with that of another participant. 
A participating agent $i$ in the peer prediction mechanism of \cite{MRZ05} is paired with a uniformly randomly selected agent $j$, and paid as a deterministic function of her own reported bit $\hat{b}_i$, as well as the reported bit $\hat{b}_j$ of agent $j$. Such a payment rule is inherently disclosive;\footnote{since it is computed deterministically as a function of the reported bit of a single individual $j$} for the payments to satisfy differential privacy, we must instead base them on {\em perturbed aggregates} of all players' reports. Relatedly, if the privacy costs $c_i$ are  drawn from a distribution with unbounded support, no finite level of payment will encourage full participation with truthful reporting, in contrast to \cite{MRZ05}. Hence, the derivation of Bayes-Nash equilibrium conditions is more delicate: we require players to be strictly incentivized to report their true bit, which can be interpreted as a prediction of the average value of all other reported bits, but this incentive must be robust to the noise that is explicitly added to the computed aggregate (which determines the payment) in order to preserve privacy, as well as to the error in the estimate from lack of full participation.
\item As mentioned above, in order to guarantee differential privacy, we require that the output (and payments) we compute be perturbed, which introduces error in our estimate of the statistic that we wish to compute. To obtain some fixed level of accuracy, therefore, we must deal with a tradeoff: on the one hand, increasing the perturbation rate directly decreases accuracy by virtue of the noise we add. On the other hand, it increases privacy, which can in turn increase participation, which \emph{increases} accuracy. In order to find an equilibrium that matches our accuracy goals, we must manage this tradeoff.
\end{enumerate}

In this paper, we give a differentially private peer prediction mechanism that supports the {\em accurate} estimation of a population statistic as a Bayes Nash equilibrium, even in settings where agents have explicit preferences for privacy. Our mechanism, which guarantees $\epsilon$-differential privacy to each agent $i$ against any adversary who can observe the statistical estimate output  by the mechanism as well as the payments made to the $n-1$ other agents $j\neq i$ (but not the payment made to agent $i$,\footnote{We assume that the surveyor is trusted, and that
agents are concerned about maintaining privacy from the other people
in the survey, as well as from external observers of the survey
outcome. We think of payments made to a particular agent as
confidential, so that an outside observer, together with a coalition
of survey participants, could in the worst case learn the outcome of
the survey, as well as the payments made to all players $j \neq i$, and use these to
form inferences about agent i's private data. (Note that if the
adversary could see the payment made to agent $i$, then it is not hard
to see that almost nothing can be done to preserve privacy: observing the payment of
player $i$ reveals that her cost for participating is lower than the
payment, which is a violation of differential privacy (of the costs)).}) requires knowledge of the marginal prior distribution on bits $b_i$, but need only know a crude upper bound on the marginal distribution on costs $c_i$. We show that with a slightly more specific assumption about agent privacy costs~\cite{CCKMV13}, the mechanism's total cost of compensating players for their privacy goes to zero as the number of agents diverges, while supporting accurate estimation of the statistic in equilibrium.


\subsection{Related Work}
There are two main strands of related work: that in the ``sensitive surveys'' literature, and that in the peer prediction literature.

\subsubsection{Privacy Aware Surveys}
The problem of estimating a population statistic among strategic agents who have explicit costs for privacy was introduced by Ghosh and Roth \cite{GR11}, who considered the problem of designing a prior-free direct revelation mechanism for this task. They considered two settings: one in which agents' costs $c_i$ can be correlated with their private bits $b_i$ (and hence the costs themselves can potentially be disclosive, and computations on them can harm privacy), and another in which agents' costs are independent from their private bits (or at least agents do not regard themselves as experiencing privacy costs due to computations done only on their reported $c_i$ values). In the first (more interesting) model, \cite{GR11} prove an impossibility result, showing that no individually rational mechanism that makes finite payments can achieve a non-trivial estimate of the underlying population statistic. In the second (less realistic) model, \cite{GR11} give dominant strategy truthful mechanisms for either a) producing a statistic to some target accuracy $\alpha$, with optimal payments relative to an envy-free benchmark, or b) producing a maximally accurate statistic given a fixed payment budget.

Much of the subsequent work has focused on circumventing and extending the impossibility result that \cite{GR11} prove in the setting in which agents' bits and costs are correlated. Fleischer and Lyu \cite{FL12} consider a Bayesian setting in which there are two known prior distributions on costs, $D_0$ and $D_1$. Agents with $b_i = 1$ have their cost drawn from $D_1$, and agents with $b_i = 0$ have their costs drawn from $D_0$. The mechanism knows the priors, but does not know the proportion of individuals with $b_i = 1$. \cite{FL12} give a clever contracting scheme that makes the distribution on agent participation decisions independent of their bit, and from this argue that they can incentivize agent participation even in the presence of correlations between bits and costs. This mechanism, however, does require that these priors be known exactly to the mechanism: both the privacy guarantee and the truthfulness guarantee depend on the accuracy of these priors.

Ligett and Roth \cite{LR12} in contrast give a model in which the impossibility result of \cite{GR11} is circumvented by relaxing the individual rationality requirement. In this model, the mechanism has the power to \emph{observe non-participation decisions}---which, when costs and bits are correlated, can result in privacy loss for the non-participating agents. However, \cite{LR12} only require that the mechanism make (individually rational) payments to participating agents, and not to non-participating ones.

Most recently, an elegant paper of Nissim, Vadhan, and Xiao \cite{NVX14} revisits this problem and offers two results. First, they significantly strengthen the impossibility result of \cite{GR11} to hold under a much weaker set of assumptions. Second, they circumvent the impossibility result by making the natural assumption that agent costs are positively correlated with their bit: that is, they make the assumption that there is a clear ``more sensitive'' bit value, such that an agent whose bit flips from $b_i = 0$ to $b_i = 1$ will only have a higher, and never lower cost for privacy. By promising privacy only for agents that have this property, they are able to circumvent the impossibility theorem.

Common to all these papers (which are surveyed in \cite{PR13}, along with a broader set of work on the intersection of privacy and mechanism design) is the assumption that agent bits $b_i$ are \emph{verifiable}, and that agents only have the ability to misreport their costs $c_i$. This represents the main departure of the present paper from previous work on the sensitive surveyor's problem: we allow agents to mis-report their bit $b_i$.

\subsubsection{Privacy Cost Functions}
Any work that seeks to model strategic agents in the presence of privacy concerns must grapple with the problem of how to model agents' costs for privacy. \emph{Differential privacy}, introduced by Dwork, McSherry, Nissim, and Smith \cite{DMNS06} provides a formal measure which gives a linear upper bound on the degree to which an agent can decrease his expected future utility via his decision to participate in a mechanism (see \cite{GR11} for further discussion of this). 
\cite{GR11} use a model that assumes agent costs are \emph{exactly} linear: that they experience cost $c_i\epsilon$ for some $c_i$ when their data are used in an $\epsilon$-differentially private way.  Nissim, Orlandi, and Smorodinsky \cite{NOS12} propose assuming that agents' privacy costs can be arbitrary, but \emph{upper bounded} by some linear cost $c_i\epsilon$. This is the most general assumption in the literature, and the one that we adopt for most of our paper.

Chen et al.~\cite{CCKMV13} propose and justify an outcome-based measure (also based on differential privacy) which usually results in much lower (sublinear) costs, which have several nice analytic properties.  We also adopt this model in Section~\ref{sec:alternative} as an alternative (and relaxed) definition of privacy.

\subsubsection{Peer Prediction}
The peer-prediction method, introduced by Miller, Resnick and Zeckhauser~\cite{MRZ05}, is a mechanism for truthfully eliciting information from agents in the absence of a verifiable ground truth against which to compare agents' reports. The mechanism in~\cite{MRZ05} uses proper scoring rules to reward agents for making reports that are predictive of other agents' reports, to ensure that truthtelling is a Nash equilibrium. There has since been a stream of work on several variants of the original peer-prediction model, motivated by opinion elicitation in online settings such as reviews  and reputations where there is no objective ground truth~\cite{jurcafaltings06,jurcafaltings07,jurcafaltings09,witkowski_aaai12,witkowski_ec12}. The Bayesian Truth Serum (BTS) mechanism and its variants~\cite{prelec2004bayesian,witkowski2012robust} use a different technique for eliciting truthful reports of unverifiable information that do not require the assumption of a common prior. Finally, there is also work focused explicitly on incentive-compatible mechanisms for surveying a population to acquire an estimate of some statistic~\cite{lambert2008truthful,jurcafb08,goelrp09,papakonstantinou2011mechanism}. The papers amongst these that are closest to our problem are perhaps~\cite{jurcafaltings06}, which uses automated mechanism design to minimize the budget required by an incentive-compatible information elicitation mechanism, albeit in a setting with equal participation costs, and~\cite{goelrp09} which presents a truthful, weighted, budget-balanced mechanism for collective revelation.

This information elicitation literature, however, studies models that differ fundamentally from our setting in that agents derive utility only from the payment they receive from the mechanism, and not from any use of the reported information itself. Specifically, this literature thus far has focused, to the best of our knowledge, entirely on settings where agents do not derive any (dis)utility from any outputs computed from the reports made to the mechanism, unlike in our setting with privacy costs (where the outputs are the payments made to (other) agents, and the aggregate statistic computed by the  analyst). Also, agents' participation costs in this literature (which in fact has largely assumed equal participation costs for all agents) do not relate to the (private) information the surveyor wishes to elicit from them. These two differences between our setting with a privacy sensitive population and the models in the information elicitation literature significantly alter both the incentives that need to be provided for truthful reporting, as well as the tradeoffs between the total cost incurred by the mechanism and the accuracy of its output.

\section{Model}
We now present a model for the problem of designing a survey to accurately aggregate {\em non-verifiable} information from a population of privacy-sensitive agents, each of whom strategically chooses whether or not to participate, and whether or not to truthfully report her private bit to the survey. \\

{\bf Agents}. There are $n$ agents in the population, each with a private bit $b_i \in \{0,1\}$ of interest to an analyst, and a private cost coefficient $c_i$ that characterizes the agent's disutility from any privacy loss from the survey. The vector of bit-cost pairs $(b,c)= [(b_1, c_1),\ldots,(b_n,c_n)]$ describing the population of $n$ agents is drawn from a known joint distribution $\mathcal{P}$ over $\{0,1\}^n \times \mathbb{R}^n_{\geq 0}$; since the prior $\mathcal{P}$ is a joint distribution over bit-cost {\em pairs}, this model allows for agents to have privacy costs $c_i$ that are correlated with their sensitive private bit $b_i$. (We note that the mechanisms studied in this paper elicit agents' {\em bits}, but do {\em not} ever require reports of the {\em costs}---the privacy costs are only used by agents in their own utility calculations to determine their optimal action trading off the costs and benefits of participation and truthful reporting.)

We make two assumptions about the joint distribution $\mathcal P$. (i) We assume that  $\mathcal P$ is {\em symmetric} over agents, in that all agents are `equal' in terms of what bit-cost pair they may draw. Formally, for every permutation $\sigma$ and every $b \in \{0,1\}^n, c \in \mathbb{R}^n_{\geq 0}$, we have $\Pr_{\mathcal{P}}[(b, c)] = \Pr_{\mathcal{P}}[\sigma((b, c))]$, \ie,  all permutations of a given bit-cost vector describing the population are equally likely.\footnote{Of course, note that this does not mean that all individual bit-cost {\em pairs} are equally likely; it only says that all {\em agents} are equally likely to have a particular pair.} (ii) Second, we assume that an agent's {\em cost} $c_i$ does not give her any information about other agents, beyond what is already conveyed to her by her own {\em bit} $b_i$. Formally, the posterior distribution on other agents' bits $b_{-i}$ and costs $c_{-i}$ that any agent $i$ can compute after observing his own type $(b_i, c_i)$ is conditionally independent of $c_i$ given $b_i$: for every $b_i, b_{-i}, c_i, c_i', c_{-i}$, $\Pr[b_{-i} | b_i, c_i] = \Pr[b_{-i} | b_i, c_i']$. This is a natural condition that can be satisfied in a number of ways---for example, it is satisfied if agents \emph{first} draw a bit $b_i$ from an underlying distribution on bits, and then draw their cost from a distribution parameterized by $b_i$, as in the model considered by \cite{FL12}.

Since an agent's cost, conditional on her bit, does not affect the distribution of remaining agents' bits, and since agents are symmetric, the distribution over the $n-1$ remaining draws of private bits $b_{-i}$ given the draw of $b_i$ does not depend either on the identity $i$ of the agent who made the draw, nor on her cost $c_i$. We will use the following notation for these common posterior distributions.
\begin{definition}\label{d-P0P1}[$\mathcal{P}_{0}, \mathcal{P}_{1}, \mathcal{B}, \mathcal{B}_{0}, \mathcal{B}_{1}, \mathcal{C}, \mathcal{C}_{0}$, $\mathcal{C}_{1}$.]
Consider a prior $\mathcal P$ over the $n$-vector of bit-cost pairs $[(b_i, c_i)]$.  We denote the conditional joint distribution of the vector $(b, c)_{-i} \in (\{0,1\} \times \mathbb{R})^{n-1}$ conditioned on $b_i = 0$ by $\mathcal{P}_{0}$ and that conditioned on $b_i = 1$ by $\mathcal{P}_{1}$.  We denote the marginal distribution of $\mathcal{P}$ (and $\mathcal{P}_{0}$ and $\mathcal{P}_{1}$) on the bits as $\mathcal{B}$  (and  $\mathcal{B}_{0}$ and $\mathcal{B}_{1}$) and the marginal distribution of the costs as $\mathcal{C}$  (and  $\mathcal{C}_{0}$ and $\mathcal{C}_{1}$)
\end{definition}

{\bf Mechanisms.} There is an analyst who would like to acquire each agent's private bit $b_i$ in order to perform some computation on it. The analyst uses a {\em mechanism} $M$ to transform the set of bits $\hat{b}_i$ reported by participating agents into an outcome $o$; in exchange, the mechanism can make \emph{payments} $\pi_i$ to the agents for their reports. Formally, a mechanism is a randomized mapping $M:\{0,1\}^n\rightarrow \mathcal{O} \times \mathbb{R}_{\geq 0}^n$, taking as input the vector of reported bits $\hat{b}_i \in \{0,1,\bot\}$ (where $\bot$ represents the decision to decline to provide one's bits), and produces an output $M(\hat{b}) = (o, \pi)$ where $o \in \mathcal{O}$ is a publicly observed outcome in some abstract outcome space $\mathcal{O}$, and $\pi \in \mathbb{R}_{\geq 0}^n$ is a vector of payments.
We will think of the analyst as interested in computing the fraction of the population that has $b_i = 1$, so that the mechanisms we design will have outcome space $\mathcal{O} = [0, 1] \subseteq \mathbb{R}$.

Unlike in prior literature, we do {\em not} assume that the bits $b_i$ are verifiable by the analyst. In our setting, in addition to declining to provide their bit to the mechanism, agents can {\em lie} about their private bits $b_i$, and will do so if it improves their utilities, which we describe next. \\

{\bf Agent utilities.} An agent's utility is the difference between the payment $\pi_i$ she receives from the analyst, and the cost she incurs from privacy losses from the outcome of the mechanism, modeled as follows.

Each agent in the population is concerned about \emph{privacy}, in the sense of what might be revealed about her sensitive bit $b_i$ by information that becomes available to others as a result of the survey. This includes both the publicly released outcome $o$, as well as the payments $\pi_{-i}$ made to the \emph{other} agents by the mechanism; we assume, however,  that the payment $\pi_i$ made to player $i$ is unobservable to anyone other than player $i$.  We note that this is the strongest privacy model possible, as it is easy to see that if an adversary  can see player $i$'s payment, then no mechanism can be both finitely differentially private and truthful in our setting.
We measure privacy loss in our setting using a variant of \emph{differential privacy} \cite{DMNS06} called \emph{joint-differential privacy}, introduced by \cite{KPRU14}:

\begin{definition}[Joint Differential Privacy] Let $M(b)_{-i} = (o, \pi_{-i})$ denote the portion of the mechanism's output that is observable to outside observers and agents $j \neq i$. A mechanism $M:\{0,1\}^n\rightarrow \mathcal{O} \times \mathbb{R}_{\geq 0}^n$ is $\epsilon$-jointly differentially private if for every vector of bits $\hat{b} \in \{0,1\}^n$ for every player $i$, for each bit $\hat{b}_i' \in \{0,1\}$, and for every observable set of outcomes $S \subseteq \mathcal{O} \times  \mathbb{R}_{\geq 0}^{n-1}$:
$$\Pr[M(\hat{b})_{-i} \in S] \leq \exp(\epsilon)\Pr[M(\hat{b}_i',\hat{b}_{-i})_{-i} \in S.]$$
\end{definition}
Note that the standard notion of differential privacy would require that a unilateral deviation by a single player $i$ would result in only a small change in the distribution over outcomes $o$ as well as payment vectors $\pi$, \emph{including} player $i$'s own payment $\pi_i$ -- meaning that each player would be paid roughly the same amount, no matter what bit they report! Such a mechanism clearly cannot incentivize truthful reporting of one's bits. Joint differential privacy relaxes this constraint, and insists only that the joint distribution over the outcome $o$ and the payments $\pi_{-i}$ made to all players $j \neq i$ be insensitive to the report of player $i$. Crucially, $\pi_i$ can depend on player $i$'s report in an arbitrary way.

Following Nissim, Orlandi, and Tennenholtz \cite{NOS12}, we model agents as having costs for privacy $v_i(M, o, \pi_{-i})$ that can be arbitrary functions of the mechanism and the output.  However, we assume that if the mechanism $M$ is jointly differentially private, the privacy cost to an agent for all possible outcomes of the mechanism is {\em upper-bounded} by a linear function of the level of differential privacy: that is, if $M$ is $\epsilon$-jointly differentially private, $v_i(M, o, \pi_{-i}) \leq \epsilon c_i$ for all outcomes $o$, where $c_i$ is agent $i$'s personal cost coefficient. When it is clear from context, we will simply write $v(\epsilon, o, \pi_{-i})$ to emphasize that the upper bound depends on the mechanism $M$ only via the level of joint differential privacy $\epsilon$  that $M$ provides, and not on any other characteristic of the mechanism.

We assume that players are risk neutral, so that a player's utility for the outcome of a mechanism $M$ that is $\epsilon$-jointly differentially private is
$$u_i(M(\hat{b})) = \E[\pi_i] - \E[v_i(M, o, \pi_{-i})] \geq \E[\pi_i] - \epsilon c_i$$

We are interested in survey-like mechanisms that ask each agent to report their private bit. An agent's action choices in the mechanisms we consider are whether or not to participate, and what bit $\hat b_i$ to report if participating; a \emph{strategy} for an agent $i$ is a function $\sigma_i$ mapping agent type $(b_i, c_i)$ to an action in $\hat{b}_i \in \{0, 1, \bot\}$ (recall that $\bot$ represents non-participation). Since we are in an incomplete information setting, with a known common prior from which agents' types are drawn, we use the solution concept of a Bayes-Nash equilibrium.

\begin{definition}
A set of strategies $\sigma_1,\ldots,\sigma_n$ forms a \emph{Bayes-Nash equilibrium} if for every player $i$, for every realized bit $b_i$ and every alternative strategy $\sigma_i'$:
$$\E_{(b, c) \sim \mathcal{P}}[u_i(M(\sigma(b, c))) | b_i] \geq \E_{(b, c) \sim \mathcal{P}}[u_i(M(\sigma_i'(b_i,c_i),\sigma_{-i}(b_{-i}, c_{-i}))) | b_i].$$
\end{definition}

Note that when we take the expectation over other player's types $b_{-i}, c_{-i}$, we condition only on $b_i$, and not on $c_i$: this is because of our assumption that the posterior distribution is conditionally independent of $c_i$ given $b_i$.


{\bf Accuracy.} The analyst is interested in computing the fraction of the population that has $b_i = 1$: we use $\hat{p} = \frac{1}{n}\sum_{i=1}^n b_i$ to denote this quantity. Since agents in the population each strategically decide whether to participate and what to report as their bits, the estimate actually computed by the analyst may be different from $\hat p$.
The analyst's goal is to design a mechanism such that its outcome closely approximates the true fraction of agents with $b_i = 1$ in the population, measured formally via its accuracy.
\begin{definition}[Accuracy.]
We say that an output estimate $\tilde{p}$ is $\alpha$-accurate with respect to the population if $|\hat{p} - \tilde{p}| \leq \alpha$.
\end{definition}

We wish to design mechanisms that are $\alpha$-accurate with high probability {\em in equilibrium}---\ie, such that the estimate $\tilde{p}$ computed by the mechanism, based on the {\em reported} bits, from the agents who choose to participate, is within an additive error $\alpha$ of the true population statistic. Note that a (joint-)differentially private mechanism $M$ cannot simply return $\tilde{p}$ as the fraction of reports $\hat b_i$ that are equal to $1$, since it must necessarily perturb this fraction to guarantee privacy.


\section{Private peer-prediction}
We now address the problem of conducting sensitive surveys in settings without verifiability: how can an analyst design payments to ensure an accurate, representative outcome when she cannot verify whether privacy-sensitive agents honestly report their private information? To do this, we build on the {\em peer prediction} mechanism~\cite{MRZ05}, designed to elicit information from strategic agents when their reports are unverifiable. The peer prediction  mechanism~\cite{MRZ05}, however, is not differentially private, nor does it consider agents who have privacy costs or preferences over outcomes. But modifying the peer-prediction mechanism to be differentially private also modifies the mechanism's incentive properties, since an agent's incentives for truthtelling now depend not only on whether a random reference agent will truthfully report her bit or not, but also on what fraction of the {\em entire} population of (strategic) potential participants enters the survey and truthfully reports their private information. In this section, we present a differentially private peer-prediction mechanism, and analyze the incentives of agents with privacy costs in this mechanism.

\subsection{Peer Prediction Preliminaries}
We begin by introducing some necessary background on scoring rules, which are designed for information elicitation in settings with {\em verifiable} outcomes.

Consider a forecasting setting, where the objective is to gather and aggregate the opinions of a set of experts about some (observable) future event. There is a large literature on the design of mechanisms---{\em scoring rules}---that reward an expert based on her reported prediction and the observed outcome of the event so as to incentivize the expert to truthfully reveal her true prediction about the event.  Specifically, suppose there is an single expert who has some private belief about the probability of a random binary event. A {\em strictly proper scoring rule} provides strict incentives to the expert to truthfully report her belief, \ie, the payment made by such a rule is such that the expert uniquely maximizes her expected payment by reporting her true belief.

In this paper, we will illustrate our approach to private peer-prediction via a particular strictly proper scoring rule, the well-known {\em Brier} scoring rule~\cite{Brier}, which makes payments for predicting a binary event as follows.
Let $I$ be the indicator random variable for the binary event to be predicted, and let $q$ be a prediction of the probability of the event occurring. The payment for prediction $q$ depends on the realized outcome $I$ as
\[
BasicBrier(I, q) = 2I\cdot q + 2(1-I)\cdot (1-q) - q^2 - (1-q)^2,
\]
and it is easy to verify that an expert who believes that $I$ occurs with probability $p$ will maximize her expected payment (with respect to her belief about the probability of $I$) by reporting $q = p$.

The following extension of the basic Brier scoring rule is central to our mechanism.

\begin{definition}[$B(p,q)$, $B_{c, d, \rho}(p, q)$]
For any $p$ and $q$, we define the payment function $B(p,q)$ as follows:
\[
B(p,q) = 1 - 2(p - 2p\cdot q + q^2).
\]
We also define a rescaling $B_{c, d, \rho}(p, q)$ of the payment function $B(p,q)$ as follows where $\rho > 0$:
\[
B_{c, d, \rho}(p, q) = \rho(B(p-c, q-c)-d).
\]
\end{definition}
The definition of $B(p,q)$ has a simple interpretation in terms of the basic Brier scoring rule---$B(p,q)$ is exactly the {\em expected} payment under the basic Brier scoring rule from reporting a guess $q$ for the probability of an event, when the expert believes that event $I$ occurs with probability $p$:
\begin{align*}
\E_{I \sim p}[BasicBrier(I,q)]&=2 p\cdot q + 2(1-p)\cdot (1-q) - q^2 - (1-q)^2\\
&= 1 - 2(p - 2p\cdot q + q^2)\\
&= B(p,q).
\end{align*}
The need for the modified payment function $B_{c, d, \rho}(p, q)$ will become clear when we move to discussing private peer prediction mechanisms. The values of $\rho$ and $d$ simply apply a linear shift to the payoff $B(p,q)$, while the effect of $c$ is slightly more subtle---we will use it when defining payments in the peer-prediction mechanism to symmetrize the scoring rule across observations $b_i = 0$ and $b_i = 1$.

The following easy facts about $B_{c, d, \rho}(p,q)$ will be useful in our equilibrium analysis.
\begin{proposition}
\label{p-Bpq}
The payment scheme $B_{c, d, \rho}(p,q)$ 
has the following properties. 
\BNUM
\item For any $p$, $B_{c, d, \rho}(p,q)$ is (uniquely) maximized by reporting $q = p$.
\item $B_{c, d, \rho}(p,q)$ is Lipschitz-continuous in $p$: For any $p$, $p'$,
\BEQ
\label{e-Bpq2}
|B_{c, d, \rho}(p,q) - B_{c, d, \rho}(p',q)| \leq \lambda|p - p'|,
\EEQ
where $\lambda = |\rho(2-4(q-c))|$.
\item Suppose the value of $p$ is randomly drawn from a distribution with mean $\E[p]$. Then an agent's expected payment from reporting $q$ is
\BEQ
\label{e-Bpq3}
\E_p[B_{c, d, \rho}(p,q)] = B_{c, d, \rho}(\E[p],q),
\EEQ
and this payment is maximized by reporting $q = \E[p]$.
\item For any $p,q$,
\BEQ
\label{e-Bpq4}
B_{c, d, \rho}(p,p) - B_{c, d, \rho}(p,q) = 2\rho(p-q)^2.
\EEQ
\ENUM
\end{proposition}
\begin{proof}
The first statement follows immediately by setting the derivative of $B_{c, d, \rho}(p,q)$  with respect to $q$ to be zero since $B_{c, d, \rho}$ is strictly concave in $q$, the second from noting that $B_{c, d, \rho}(p,q) - B_{c, d, \rho}(p + \epsilon, q) = \rho(2-4(q-c))\epsilon$, and the third from noting that $B(p,q)$ is linear in $p$ so that $\E_p[B(p,q)] = B(\E[p],q)$, and then applying the first statement. The last statement is obtained by direct substitution:
\begin{align*}
B_{c, d, \rho}(p,p) - B_{c, d, \rho}(p,q) = ~& \rho(1-2(p-c) + 4(p-c)^2 - 2(p-c)^2-d) \\
&~-\rho(1 -2(p-c)+4(p-c)(q-c) -2(q-c)^2 -d) \\
=~&  2\rho(p-q)^2.\qed \\
\end{align*}
\end{proof}

\noindent {\bf Peer prediction.} In our setting, there is no publicly observable random event or ``ground truth'' whose outcome is being predicted by an ``expert''. Instead, we only have reports from the agents being surveyed about their private bits, albeit drawn from a common prior distribution. The key idea behind the {\em peer-prediction} mechanism~\cite{MRZ05} designed for such information elicitation problems with {\em unverifiable information} is the following: an agent's report, instead of being rewarded for its ability to predict an observable event, can be rewarded for its success in ``predicting''  the outcome of the random event consisting of another agent's draw of her private bit. Note that the mechanism's only access to this private bit is via the agent's report, which necessitates introducing the notion of {\em equilibrium} truthful reporting: since each agent's payoff, which depends on the ability of her report to predict another agent's private bit, depends on that agent's report, we cannot just ask what a particular agent will report (as with scoring rules); rather, we must ask what reporting strategies of the population of agents constitute an equilibrium.

The vector of private bit-cost pairs $[(b_i, c_i)]$ of the $n$ agents in our population is drawn from a (known) prior distribution $\pd$;
recall from Definition \ref{d-P0P1} that $\mathcal{P}_{0}$ (respectively $\mathcal{P}_{1}$) denotes the conditional distribution of the vector $b_{-i} \in \{0,1\}^{n-1}$ conditioned on $b_i = 0$ (respectively $b_i = 1$). Since these distributions $\mathcal{B}_{0}$  and $\mathcal{B}_{1}$ are known (they can be computed from the known prior $\pd$), an agent $i$'s draw of $b_i$ can be used to compute a posterior probability that a random agent $j\neq i$ draws $b_j = 1$.
The fact that a particular draw of $b_i$ leads to an updated value for $\Pr(b_j = 1)$ means that one can ask agents for their bit $b_i$, and use it to compute their prediction of the probability that $b_j = 1$ using $\pd$, instead of asking them to directly report their prediction of this probability.\footnote{Note that eliciting $b_i$ is not, in general, equivalent to eliciting the probability, since the set of strategies available to an agent facing a scoring rule of the form $B(p,q)$ shrinks when reporting a single bit---while in general an agent could choose to report any $q \in [0,1]$, asking an agent to report a bit $\hbi$ restricts her to only two possible values of $q$, $p_1$ or $p_0$, corresponding to the reports of $\hbi = b_i$ and $\hbi = 1-b_i$. } Thus, if agent $i$ reports $b_i = 0$, then the mechanism can translate that into the prediction $p_0 = \Pr[b_j = 1| b_i=0]$ and if agent $i$ reports $b_i = 1$, then the mechanism can translate that into the prediction $p_1 = \Pr[b_j = 1| b_i=1]$.

A peer prediction mechanism based on the Brier scoring rule would then pay agent $i$ $BasicBrier(\hat{b}_j,p_{\hat{b}_i})$ where $\hat{b}_i$, and $\hat{b}_j$ are the reports of agents $i$ and $j$ respectively. The first claim in Proposition \ref{p-Bpq}, applied with $c = d = 0$ and $\rho = 1$, immediately yields that truthful reporting is a Nash equilibrium with these payments: consider an agent $i$ with reference agent $j$. If $j$ reports her bit truthfully, \ie, $\hat{b}_j = b_j$, then $i$'s belief about the probability of $j$'s report being $1$, conditional on her draw, is exactly $\Pr[b_j = 1| b_i]$. Proposition \ref{p-Bpq}(i) says that $i$ maximizes her payoff in the payment rule $BasicBrier$ by reporting her true belief about this probability, namely $p_{b_i}$, which corresponds precisely to reporting her true bit $b_i$ to the mechanism.

For a peer-prediction mechanism to guarantee differential privacy to agents, however, the mechanism must introduce a  perturbation into the reports that are used to compute other agents' payments (recall that our notion of privacy assumes that all payments, except those made to agent $i$ herself, are publicly visible, and therefore may reveal information about agent $i$. But such a perturbation will change the quantity that agent $i$'s prediction $p_{\hat{b}_i}$ is being compared against, and correspondingly her payoffs, and possibly her incentives.  We will therefore need to construct an analog of Proposition \ref{p-Bpq}(i) that will allow us to declare that an agent will do better by truthfully reporting her bit $b_i$ than by lying, assuming other agents also report truthfully, even when the bit he is trying to predict is not actually drawn according to the distribution $\Pr[b_j = 1| b_i=0]$ but only a distribution close (enough) to it.

\vspace{5pt}

The following proposition creates a peer-prediction mechanism that is resilient to noise, which forms the basis for designing a differentially private peer-prediction mechanism with truthful reporting in equilibrium. The first part of the  proposition provides an {\em upper} bound on an agent's payoff from {\em lying} about her bit, assuming all other agents report truthfully, while the second provides a {\em lower} bound on the payoff from {\em truthful} reporting, as a function of the magnitude of the noise. Together, these two statements will allow us to lower-bound the increase in payoff from truthful reporting over lying (assuming other agents report truthfully) as a function of how much noise is added to the computation by a differentially private mechanism; we will then compare these against the improvement in privacy cost from lying to investigate when truthful reporting constitutes an equilibrium for privacy-sensitive agents, in Theorem \ref{t-mptp}.

\begin{proposition}~\label{prop:pp-values}
Consider a prior $\pd$, and let $p_0, p_1 \in [0, 1]$ be the posterior values $\Pr[b_j = 1| b_i]$ corresponding to this prior for $b_i = 0$ and $b_i = 1$ respectively. Let $\alpha, \beta \in \mathbb{R}$ be values such that $\beta > 0$ and $\alpha < \frac{|\po-\pz|}{2}$, and set the scale-and-shift parameters $c,d,\rho$ in the modified Brier scoring rule $B_{c, d, \rho}$ to be
\begin{align*}
c &= (p_0 + p_1 - 1)/2,\\
d &= \frac{1}{2} - \frac{3}{2}(\po-\pz)^2 + 2\alpha|\pz - \po|,\\
\rho &= \frac{\beta}{2(\po-\pz)^2-4\alpha|\pz - \po|}.
\end{align*}
Then, (i) $\rho > 0$, and (ii) the following inequalities hold for each of $b \in \{0, 1\}$:
\BNUM
\item For any $p'$ with $|p' - p_b| < \alpha, \:\:\:   B_{c, d, \rho}(p', p_{1-b}) \leq 0$.  \label{eq:lyingwithtolerance}
\item For any $p'$ with $|p' - p_b| < \alpha, \:\:\: \beta \leq B_{c, d, \rho}(p', p_b)$. \label{eq:truthwithtolerance}
\item For any $\alpha' > 0$, $p'$ with $|p' - p_b| < \alpha', \:\:\: \beta + 2 \rho (\alpha + \alpha')|\pz - \po| \geq B_{c, d, \rho}(p', p_b)$. \label{eq:paymentbound}
\ENUM

\end{proposition}

\begin{proof}
Note that $\rho > 0$ follows from $\beta > 0$ and $\alpha < \frac{|\po-\pz|}{2}$.
To prove the next set of inequalities, first note that $B(1/2 + p, 1/2 + q) = B(1/2 - p, 1/2 - q)$; this can be confirmed by algebraic manipulation.  Thus we have
\begin{align*}
B(p_0 -c, p_0 -c) &= B\left(\frac{1}{2} + \frac{p_0 - p_1}{2},\frac{1}{2} + \frac{p_0 - p_1}{2}\right)  \\
&=  B\left(\frac{1}{2} + \frac{p_1 - p_0}{2},\frac{1}{2} + \frac{p_1 - p_0}{2}\right) \\
&= B(p_1 -c, p_1 -c),
\end{align*}
and similarly
\begin{align*}
B(p_0 -c, p_1 -c) &= B\left(\frac{1}{2} + \frac{p_0 - p_1}{2},\frac{1}{2} + \frac{p_1 - p_0}{2}\right)  \\
&=  B\left(\frac{1}{2} + \frac{p_1 - p_0}{2},\frac{1}{2} + \frac{p_0 - p_1}{2}\right) \\
&= B(p_1 -c, p_0 -c).
\end{align*}

Next, note that for $b \in \{0, 1\}$,
\BEQ
 B(p_b - c, p_{b} - c ) = B\left(\frac{1}{2} + \frac{p_b - p_{b-1}}{2},\frac{1}{2} + \frac{p_b - p_{b-1}}{2}\right)  = \frac{1}{2} + \frac{1}{2}(\po-\pz)^2, \label{eq:briertruth}
\EEQ
where the last equality follows from the definition of $B(p,q)$. Similarly, for $b \in \{0, 1\}$,
\BEQ
 \label{eq:brierlying}
 B(p_b - c, p_{1-b} - c ) = B\left(\frac{1}{2} + \frac{p_b - p_{b-1}}{2},\frac{1}{2} + \frac{p_{1-b} - p_{b}}{2}\right)
=  \frac{1}{2} - \frac{3}{2}(\po-\pz)^2.
\EEQ

Next, we compute the expected payoff $B_{c, d, \rho}(p_b, p_b)$ to an agent for truthfully her private bit, {\em assuming} all other agents report truthfully, for both possible values of the private bit $b \in \{0, 1\}$:
\begin{align}
B_{c, d, \rho}(p_b, p_b) & = \rho(B(p_b - c, p_{b} - c ) - d) \nonumber\\
                         & =  \rho\left(\frac{1}{2} + \frac{1}{2}(\po-\pz)^2 - \left(\frac{1}{2}  - \frac{3}{2}(\po-\pz)^2 + 2\alpha|\pz - \po|\right)\right)\nonumber\\
                         & =  \rho\left(2(\po-\pz)^2  - 2\alpha|\pz - \po|\right) \nonumber\\
                         & =  \rho(2\alpha|\pz - \po|) + \rho\left(2(\po-\pz)^2 - 4\alpha|\pz - \po|\right)\nonumber\\
                         & =  \rho(2\alpha|\pz - \po|) + \beta. \label{eq:truth}
\end{align}

Similarly, the expected payoff $B_{c, d, \rho}(p_b, p_{1-b})$ to an agent for lying (\ie, reporting $1-b$ when her true bit is $b \in \{0, 1\}$) {\em when} all other agents report truthfully, is:
\begin{align}
B_{c, d, \rho}(p_b, p_{1-b}) & = \rho\left(B(p_b - c, p_{1-b} - c ) - d\right)\nonumber\\
                         & =  \rho\left(\frac{1}{2} - \frac{3}{2}(\po-\pz)^2 - \left(\frac{1}{2}  - \frac{3}{2}(\po-\pz)^2 + 2\alpha|\pz - \po|\right)\right)\nonumber\\
                         & =  - \rho(2\alpha|\pz - \po|). \label{eq:lying}
\end{align}
To obtain the first statement (\ref{eq:lyingwithtolerance}), which upper-bounds the expected payoff from lying when an agent's report is compared against a noisy perturbation of other agents' true reports, we apply the Lipschitz condition of Proposition~\ref{p-Bpq} Part~\ref{e-Bpq2} with $\lambda = |\rho(2-4(p_b - c))| = 2 \rho |p_1-p_0|$ and $|q - p_b| \leq \alpha$ to Equation~\ref{eq:lying}. Similarly, the second statement (\ref{eq:truthwithtolerance}), which lower-bounds the expected payoff from truthtelling under noise, follows from applying the same Lipschitz condition to Equation~\ref{eq:truth}. Finally, the third statement (\ref{eq:paymentbound}) follows from applying Proposition~\ref{p-Bpq} part~\ref{e-Bpq2} with $\lambda = |\rho(2-4(p_b - c))|= 2 \rho |p_1-p_0|$ and $|q - p_b| \leq \alpha'$ to Equation~\ref{eq:truth}.
 \end{proof}


\subsection{A differentially-private peer prediction mechanism}
\label{s-diffppp}


We now address the question of designing incentives for a survey where privacy-sensitive individuals might incur disutility from the use of their data in this computation.


For a distribution $\mathcal{P}$ over types $(b, c)$, recall that we have defined posterior distributions $P_0$ and $P_1$, conditioned on a draw of $b_i = 0$ and $b_i = 1$ respectively.

Our mechanism follows. Conceptually, the mechanism asks each agent to report their private bit, which they have the option of misrepresenting. From their report, the mechanism computes the posterior belief that is consistent with their report (assuming truthful reporting). It then computes the average value of all agents' reports (treating players who have opted not to participate identically as if they reported $\hbi = 0$), and perturbs this value so as to guarantee differential privacy. Finally, it pays each agent using our modified Brier scoring rule, as if they are using the computed posterior distribution to bet on the perturbed average bit value. The payments are carefully scaled to implement a Bayes Nash equilibrium in which almost all players choose to report their bit truthfully. \\

\noindent{\bf Mechanism $\Mpna$.} Consider the following mechanism $\Mpna$ for conducting a sensitive survey with prior $\mathcal{P}$, parameterized by a participation goal $1-\alpha$, surplus payment $\beta$, and noise parameter $\eps$ with $n$ agents.
\BNUM
\item Each participating agent $i$ submits a report $\hbi$ of her private bit.
\item Set $\hbi = 0$ for each non-participating agent $i$.
\item Compute $\hat{b} = \sum_{i=1}^n \hbi$.
\item Perturb $\hat{b}$ as follows: $\bar{b} = \hat{b} + \mathrm{Lap}\left(\frac{1}{\epsilon}\right)$.
\item \[
\tp = \argmin_{x \in [0,1]}\left|x - \frac{\bar{b}}{n}\right|
\]
\[
\tp_{-i} = \argmin_{x \in [0,1]}\left|x - \frac{\bar{b} - \hbi}{n-1}\right|
\]
\item Compute $p_{\hbi}$ as follows: $ \pz = \E_{\mathcal{P}}[\tp_{-i}| b_i=0]$; and $\po = \E_{\mathcal{P}}[\tp_{-i}| b_i=1]$. Here, we take the expectation over the draws of other agents' bits, assuming they are using the strategy of reporting their bits truthfully.
\item Next compute $c$, $d$, and $\rho$ as functions of $p_0$, $p_1$, $\alpha$ and $\beta$ as in Proposition~\ref{prop:pp-values}.
\item Pay each participating agent $i$,  $\pi_i$, based on her report $\hbi$ and $\tp_{-i}$:
\[
\pi_i =  B_{c, d, \rho}(\tp_{-i}, \phbi) = \rho(1 - 2((\tp_{-i} - c) - 2(\tp_{-i}-c)\cdot (\phbi-c) + (\phbi-c)^2) -d).
\]
(Non-participating agents receive no payment.)
\item  Output estimate $\tp$.
\ENUM 

\subsection{Analyzing $\Mpna$}
In this section, we prove that $\Mpna$ is an $\eps$-joint  differentially private mechanism supporting truthful reporting by a fraction $1-\alpha$ of agents in equilibrium, and computes an accurate outcome.

We first show in Theorem \ref{T-diffp} that $\Mpna$ is $\epsilon$-jointly differentially private, which follows from routine arguments and the following lemma.
\begin{lemma}[``Billboard Lemma'' \cite{HHRRW14,RR14}]
\label{lem:billboard}
Fix any mechanism $\mathcal{M}:\mathcal{T}^n\rightarrow \mathcal{O}$, for arbitrary sets $\mathcal{T}$ and $\mathcal{O}$. Fix any function $f:\mathcal{T}\rightarrow \mathcal{O'}$. If $\mathcal{M}$ is $\epsilon$-differentially private, then the mechanism $\mathcal{M}':\mathcal{T}^n\rightarrow \mathcal{O} \times \mathcal{O'}^n$ that computes $o = \mathcal{M}(t)$ and outputs $\mathcal{M}'(t) = (o, (f(t_1, o),\ldots,f(t_n,o)))$ is $\epsilon$-jointly differentially private.
\end{lemma}

To prove joint differential privacy, we first observe that the output $\tp$ of the mechanism observed by anyone external to the mechanism is $\epsilon$-differentially private, and together with the output, the vector of payments $\pi$ is $\epsilon$-jointly differentially private.

\begin{theorem}
\label{T-diffp}
$\Mpna:\{0,1,\bot\}^n\rightarrow \mathbb{R}_{\geq 0} \times \mathbb{R}_{\geq 0}^n$ is $\epsilon$-jointly differentially private.
\end{theorem}
\begin{proof}
Privacy follows from the fact that $\hat{b}$ has sensitivity $1$, and $\bar{b}$ is computed using the Laplace mechanism of \cite{DMNS06} with scale $1/\epsilon$. From \cite{DMNS06}, we know that the computation of $\bar{b}$ is $\epsilon$-differentially private. This computation is followed by data independent post-processing to compute $\tp$, which cannot degrade the differential privacy guarantee. Finally, payments $\pi_i$ are computed as a function $f(\hat{b}_i, \bar{b}) = B_{c, d, \rho}(\tp_{-i}, p_{\hat{b}_i})$, as allowed by Lemma \ref{lem:billboard}.
\end{proof}

We next prove Theorem \ref{t-mptp}, which addresses the question of accurate estimation of the population statistic in equilibrium by our mechanism $\Mpna$: we will show that there is a {\em threshold strategy equilibrium} in $\Mpna$---where all agents with cost $c_i$ below some threshold $\tau$ participate and truthfully report their bits---that has high accuracy; this accuracy, of course, is a function of the parameters of the mechanism.

Recall that $\mathcal{P}$ is a joint distribution both over bits $b \in \{0,1\}^n$ and cost vectors $c \in \mathbb{R}^n$, and that $\mathcal{C}$ denotes the marginal distribution of $\mathcal{P}$ over cost vectors $c \in \mathbb{R}^n$. Also, $\mathcal{C}_0$ denotes the marginal distribution on cost vectors drawn from $\mathcal{P}_0$, and $\mathcal{C}_1$ the marginal distribution on cost vectors drawn from $\mathcal{P}_1$ (Definition \ref{d-P0P1}); recall that these distributions are symmetric.

We now define a quantity $\tau_{\alpha,\delta}$, which represents a cost threshold that satisfies two conditions. First, it represents a threshold such that with high probability $1-\delta$ (with respect to the prior $\mathcal{P}$), at least a $1-\alpha$ fraction of individuals in the population have costs $c_i \leq \tau_{\alpha,\delta}$. Second, it is \emph{also} a threshold such that for every player $i$, conditioned on seeing either the bit $b_i = 1$ or $b_i = 0$, the probability that a random other agent $j \neq i$ has cost $c_j \leq \tau_{\alpha,\delta}$ is at least $1-\alpha$.
\begin{definition}[Threshold $\tau_{\alpha,\delta}$.]
Fix a prior $\mathcal{P}$ with a corresponding marginal cost distribution $\mathcal{C}$, and let
$$\tau^1_{\alpha,\delta} = \inf_{\tau}\left(\Pr_{c \sim \mathcal{C}}\left[|\{i : c_i \leq \tau\}| \geq (1-\alpha)n\right] \geq 1-\delta\right),$$
$$\tau^2_{\alpha} = \inf_{\tau}\left(\min\left(\Pr_{c \sim \mathcal{C}_0, j \neq i}[c_j \leq \tau], \Pr_{c \sim \mathcal{C}_1, j \neq i}[c_j \leq \tau]\right) \geq 1-\alpha\right).$$
 We define $\tau_{\alpha,\delta}$ as the larger of these two thresholds:
$$\tau^1_{\alpha,\delta} = \max(\tau^1_{\alpha,\delta}, \tau^2_{\alpha}).$$
\end{definition}

We present the following theorem.
\begin{theorem}
\label{t-mptp}
Let $\mathcal{P}$  be a symmetric prior over types $(b,c)$ satisfying the condition that for every $i$, the posterior distribution $\mathcal{P}_{b_i}$ on $(b, c)_{-i}$ given $b_i$ is conditionally independent of $c_i$. Fix a participation goal $1-\alpha$ such that $\alpha < \frac{|\po-\pz|}{2}$, a privacy parameter $\epsilon$, and the desired confidence $\delta$ on the accuracy guarantee. If the parameter $\beta$ in mechanism $\Mpna$ is chosen to be $\beta = \epsilon \tau_{\alpha, \delta/2}$, then:
\begin{enumerate}
\item The mechanism $\Mpna$ has a symmetric Bayes Nash equilibrium consisting of threshold strategies $\sigma_{\tau_{\alpha, \delta/2}}$, where
$\sigma_{\tau_{\alpha, \delta/2}}(b_i, c_i) = b_i$ whenever $c_i \leq \tau_{\alpha, \delta/2}$, \ie, all agents with cost smaller than $\tau_{\alpha, \delta/2}$ participate and truthfully report their private bit ($\sigma_{\tau_{\alpha, \delta/2}}$ can be arbitrary for agents with $c_i >\tau_{\alpha, \delta/2}$.)
\item In the equilibrium $\sigma_{\tau_{\alpha, \delta/2}}$, the estimate $\tp$ computed by $\Mpna$ is $\alpha'$-accurate for $\alpha' = \left(\frac{\ln(\delta/2)}{\epsilon n} + \alpha\right)$ with probability at least $1-\delta$, where the probability is over the draw of types from $\mathcal{P}$.
\end{enumerate}
\end{theorem}

\begin{remark}
If we want a mechanism which is $2\alpha$ accurate, it suffices to take $\epsilon = \frac{\ln(1/\delta)}{\alpha n}$, in which case our $\beta$ parameter is set as:
$$\beta = \frac{\ln(1/\delta)\cdot \tau_{\alpha, \delta/2}}{\alpha n}$$
\end{remark}

\begin{proof}
We first show (1), that the threshold strategies $\sigma_{\tau_{\alpha, \delta/2}}$ form a Bayes Nash equilibrium.

An agent $i$'s payoff if she participates in $\Mpna$ is the difference between her payment $\pi_i$ and her (dis)utility from her privacy loss, $v_i(\epsilon,o)$. To analyze agents' decisions in $\Mpna$, note first that
every agent's payoff depends on the decisions of all other agents in {\em two} ways---the payment $\pi_i = B_{c, d, \rho}(\tp_{-i},\phbi)$ uses the reports $\hat{b}_{-i}$ of all other agents in computing $\tp_{-i}$, and the cost $v_i(\epsilon,o)$ depends on $o$ which in turn is computed using the reports of all other agents.

Suppose all the agents other than agent $i$ are all following a symmetric \emph{threshold} strategy $\sigma_{j,\tau_{\alpha, \delta/2}}$  that results in truthful reporting of $\sigma_{j,\tau_{\alpha, \delta/2}}(b_j,c_j) = b_j$ whenever $c_j\leq\tau_{\alpha, \delta/2}$, and can result in arbitrary behavior otherwise.


We can now consider whether the same threshold strategy  $\sigma_{i,\tau_{\alpha, \delta/2}}$ is a best response for player $i$: if it is, then we will have shown that $\sigma_{\tau_{\alpha, \delta/2}}$ forms a symmetric Bayes Nash equilibrium.  Assume that $c_i < \tau_{\alpha, \delta/2}$ as otherwise there is nothing to show, since the strategy does not specify agents' actions if their cost exceeds the threshold.

Consider agent $i$'s incentives, given reports $\hat b_{-i}$ from the remaining participants in the mechanism. Agent $i$ chooses her strategy from three possible options: $\{b_i, 1-b_i, \bot\}$---truth-telling, lying, or not participating---by evaluating her expected utility given the strategy choices of the remaining agents, where the expectation is over the remaining agents' draws of bit, cost pairs $(b, c)_{-i}$ from $\mathcal P$, as well as the random coin tosses in $\Mpna$. Agent $i$ has an incentive to report the truth in $\Mpna$ if her expected payoff from reporting $\hbi = b_i$ is at least as large as the payoff from either reporting $\hbi = 1-b_i$ or not participating.

Let $S_{-i}$ denote the individuals, other than $i$,  who truthfully participate when using strategy $\sigma_{\tau_{\alpha, \delta/2}}$. Note that by the definition of $\sigma_{\tau_{\alpha, \delta/2}}$ we have  $\E[|S_{-i}|] \geq (1-\alpha) (n-1)$, since each individual truthfully reports with probability (over his type distribution) at least $1-\alpha$.   Thus for $b_i \in \{0, 1\}$ we have $$|\E[\tp_{-i} | b_i] - \E[\pbi |b_i]| \leq \alpha.$$
\ar{Not sure why we can ignore the Laplace noise...}

We first compute the payoff for truth-telling, $\E[B_{c,d,\rho}(\tp_{-1}, \pbi)|b_i] - \E_{o \sim M(b_i, b_{-i})}[v_i(\epsilon,o)]$ and show that it is positive.
Above we saw that $|\E[\tp_{-i} | b_i] - \E[\pbi |b_i]|\leq \alpha$.  Thus, by our choice of $c, d, \rho$ and Proposition~\ref{prop:pp-values} we have that  $$\E[B_{c,d,\rho}(\tp_{-i}, \pbi)|b_i] \geq \beta = \epsilon \tau_{\alpha, \delta/2}.$$
However, we also know that for each player $i$, $\E_{o \sim M(b_i, b_{-i})}[v_i(\epsilon,o)] \leq \eps c_i$ by the $\epsilon$-joint differential-privacy of the mechanism.  By assumption $\tau_{\alpha, \delta/2} > c_i$, and so it follows that truth-telling has a positive expected payoff for agent $i$.

We now show that the expected payoff for lying, $ \E[B_{c,d,\rho}(\tp_{-i}, p_{1-b_i})|b_i] - \E_{o \sim M(1-b_i, b_{-i})}[v_i(\epsilon,o)]$ is always non-positive.   Above we saw that $|\E[\tp_{-i} | b_i] - \E[\pbi |b_i]| \leq \alpha$.  Thus, by our choice of $c, d, \rho$ and Proposition~\ref{prop:pp-values} we have that  $\E[B_{c,d,\rho}(\tp_{-i}, p_{1-b_i})|b_i] \leq 0.$ The second term is non-positive, so it follows that the payoff for lying is also non-positive.

Finally, we look at the payoff for non-participation, which is simply  $- \E_{o \sim M(\bot, \hat b_{-i})}[v_i(\epsilon,o)]$ and non-positive by definition.

Because the utility agent $i$ receives from for truth-telling is nonnegative, but her utility from lying or non-participation is non-positive, she is always at least as well off truth-telling.



%
%


%

It remains to show (2), that with probability at least $1-\delta$ the mechanism is $\alpha'$-accurate for $\alpha' = \left(\frac{\ln(\delta/2)}{\epsilon n} + \alpha\right)$-when all players play according to equilibrium $\sigma_{\tau_{\alpha, \delta/2}}$. To show this, we note that error comes from two sources: first, non-truthful behavior. However, by the definition of $\sigma_{\tau_{\alpha, \delta/2}}$, we know that with probability at least $1-\delta/2$, all but $\alpha n/2$ people truth-tell in equilibrium: that is, except with probability $1-\delta/2$, we have $|\sum_{i=1}^n\hat{b}_i - \sum_{i=1}^n b_i| \leq \frac{\alpha n}{2}$. The second source of error is the Laplace noise added to $\hat{b}$. We have that $\Pr[|\mathrm{Lap}(1/\epsilon)| \geq t/\epsilon] = \exp(-t)$. Hence, $\Pr[|\hat{b}-\bar{b}| \geq \ln(\delta/2)/\epsilon] \leq \delta/2$. Combining these two bounds gives the claim.
\end{proof}

\begin{remark}  We remark that the above analysis easily extends to the case where  privacy-sensitive individuals might incur disutility from the outcome $o$ of the analyst's computation, as well as from the use of their data in this computation.  In this case, assume that an agent's cost $c_i$ bounds both the cost of his privacy and the difference in the agent's utility between any two outcomes.   Because the mechanism is $\epsilon$ private, each agent has at most $\epsilon$ influence on its outcome.  Hence the cost in the outcome quality for changing her bit (or not participating) is at most $c_i \epsilon$:  $\E_{o \sim M(b_i, b_{-i})}[v_i(\epsilon,o)]- \E_{o \sim M(b_{i}', b_{-i})}[v_i(\epsilon,o)] \leq \epsilon (\max_{o, o'} v_i(\epsilon,o) - v_i(\epsilon,o')) \leq c_i \epsilon$.
\end{remark}


\subsection{An alternative model of privacy costs}
\label{sec:alternative}

In the previous sections, we have focused on a conservative model of agents' costs for particular outcomes and privacy which assumes only that agent costs are upper bounded by a linear function of the differential privacy guarantee. This model was proposed by \cite{NOS12}, and is the weakest assumption seen in the differential privacy and game theory literature. In this section, we briefly consider a less-conservative, but still well-motivated, assumption on privacy cost functions proposed by Chen et al. \cite{CCKMV13}.  Under this stronger assumption, we are able to show not only that it is possible to truthfully elicit agents' private bits and compute on them accurately, but that the privacy cost of doing so tends to $0$ as $n$ grows large.
In fact, because non-private peer-prediction payments can be arbitrarily rescaled, the Chen et al. model implies that by gathering more participants, a surveyor could drive the entire cost of running our private peer prediction mechanism to zero.\footnote{This model does not incorporate a minimum base cost of approaching each potential survey participant.}

The following assumption is designed to capture the intuition that an agent should experience low privacy cost for outcomes that would induce only a small change in a Bayesian adversary's beliefs about her type (see \cite{CCKMV13} for a thorough motivation).
\begin{assumption}[\cite{CCKMV13} Privacy Cost Assumption]
\label{assumption}
We assume that for any mechanism $M:\{0,1\}^n\rightarrow \mathcal{O} \times \mathbb{R}_{\geq 0}^n$ , $\forall \hat{b} \in \{0,1\}^n$, for all players $i$, $o \in \mathcal{O}$, $\pi_{-i} \in \mathbb{R}_{\geq 0}^{n-1}$:
\[ |v_i(M, (o, \pi_{-i}), \hat{b}_i, \hat{b}_{-i}) |\leq c_i \ln \left( \max_{b'_i, b''_i \in \{0,1\}} \frac{\Pr[M(b'_i, \hat{b}_{-i}) = (o, \pi_{-i})]}{\Pr[M(b''_i, \hat{b}_{-i}) =(o, \pi_{-i})]} \right).\]
where
the probabilities are taken over the random choices of $M$.\footnote{The model proposed in \cite{CCKMV13} is somewhat more general, and allows the costs to be bounded by arbitrary functions of the log probability ratio, not just linear functions. We adopt this linear model here for simplicity, so that we do not need to modify our model of having a distribution of costs $c_i$.}
\end{assumption}

They show a useful lemma bounding the expected privacy loss of agents participating in a differentially private computation who experience privacy costs according to this assumption:
\begin{lemma}[\cite{CCKMV13} Composition Lemma]
\label{lem:chen}
In settings that obey the above cost assumption, and for mechanisms $M$ that are $\epsilon$-differentially private for $\epsilon \leq 1$, then for a player $i$ with bit $b_i$, $\forall b_{-i} \in \{0,1\}^{n-1}$, $\forall b'_i \in \{0,1\}$,
\[\E[v_i(M, M(b), b_i, b_{-i})] - \E[v_i(M, M(b'_i, b_{-i}), b_i, b_{-i})] \leq 2 c_i \epsilon (e^{\epsilon} - 1) \leq 4c_i\epsilon^2\]
\end{lemma}

With the above assumption, we can achieve an accurate Bayes-Nash equilibrium while scaling payments down with $\epsilon$ at a rate of $\epsilon^2$, rather than linearly. As we will see, this will allow us to drive our total costs down to zero.

We first show an analogue of Theorem \ref{t-mptp}. The only difference is that with our new stronger assumption on agent privacy costs, we can obtain a Bayes Nash equilibrium while taking the $\beta$ parameter to be substantially smaller than before.

\begin{theorem}
Suppose agents' privacy costs satisfy Assumption \ref{assumption}. Let $\mathcal{P}$  be a symmetric prior over types $(b,c)$ satisfying the condition that for every $i$, the posterior distribution $\mathcal{P}_{b_i}$ on $(b, c)_{-i}$ given $b_i$ is conditionally independent of $c_i$. Fix a participation goal $1-\alpha$ such that $\alpha < \frac{|\po-\pz|}{2}$, a privacy parameter $\epsilon$, and the desired confidence $\delta$ on the accuracy guarantee. If the parameter $\beta$ of the mechanism is chosen to be $\beta = 4\epsilon^2 \tau_{\alpha, \delta/2}$, then:
\begin{enumerate}
\item The mechanism $\Mpna$ has a symmetric Bayes Nash equilibrium consisting of threshold strategies $\sigma_{\tau_{\alpha, \delta/2}}$, where
$\sigma_{\tau_{\alpha, \delta/2}}(b_i, c_i) = b_i$ whenever $c_i \leq \tau_{\alpha, \delta/2}$, \ie, all agents with cost smaller than $\tau_{\alpha, \delta/2}$ participate and truthfully report their private bit ($\sigma_{\tau_{\alpha, \delta/2}}$ can be arbitrary for agents with $c_i >\tau_{\alpha, \delta/2}$.)
\item In the equilibrium $\sigma_{\tau_{\alpha, \delta/2}}$, the estimate $\tp$ computed by $\Mpna$ is $\alpha'$-accurate for $\alpha' = \left(\frac{\ln(\delta/2)}{\epsilon n} + \alpha\right)$ with probability at least $1-\delta$, where the probability is over the draw of types from $\mathcal{P}$.
\end{enumerate}
\end{theorem}

\begin{remark}
Note that this theorem differs from our previous theorem in how it sets $\beta$. Since we can again take $\epsilon = \frac{\ln(1/\delta)}{\alpha n}$, we have
$$\beta = \frac{4\ln(1/\delta)^2\cdot \tau_{\alpha, \delta/2}}{\alpha^2 n^2}$$
\end{remark}

\begin{proof}
We first show (1), that the threshold strategies $\sigma_{\tau_{\alpha, \delta/2}}$ form a Bayes Nash equilibrium.

Suppose the agents $j \neq i$ are all following a symmetric \emph{threshold} strategy $\sigma_{j,\tau_{\alpha, \delta/2}}$  that results in truthful reporting of $\sigma_{j,\tau_{\alpha, \delta/2}}(b_j,c_j) = b_j$ whenever $c_j\leq\tau_{\alpha, \delta/2}$, and can result in arbitrary behavior otherwise.

We can now consider whether the same threshold strategy  $\sigma_{i,\tau_{\alpha, \delta/2}}$ is a best response for player $i$: if it is, then we will have shown that $\sigma_{\tau_{\alpha, \delta/2}}$ forms a symmetric Bayes Nash equilibrium.  Assume that $c_i < \tau_{\alpha, \delta/2}$ as otherwise there is nothing to show.

An agent's $i$'s payoff if she participates in $\Mpna$ is the difference between her payment $\pi_i$ and her (dis)utility from her privacy loss, $v_i(\epsilon,o)$. We will first show that an agent's utility for participating and truth-telling is higher than her payment for participating and misrepresenting her bit.

An agent $i$ has a higher payoff for truth-telling than for lying if:
$$\E[B_{c,d,\rho}(\tp_{-i}, \pbi)|b_i]- \E[v_i(M, M(\hat{b}), \hat{b}_i, \hat{b}_{-i})] \geq \E[B_{c,d,\rho}(\tp_{-i}, p_{b_{1-i}})|b_i]- \E[v_i(M, M(\hat{b}), 1-\hat{b}_i, \hat{b}_{-i})]$$

Or equivalently,
$$\E[B_{c,d,\rho}(\tp_{-i}, \pbi)|b_i]- \E[B_{c,d,\rho}(\tp_{-i}, p_{b_{1-i}})|b_i] \geq  \E[v_i(M, M(\hat{b}), \hat{b}_i, \hat{b}_{-i})] - \E[v_i(M, M(\hat{b}), 1-\hat{b}_i, \hat{b}_{-i})].$$

By Assumption \ref{assumption} and Lemma \ref{lem:chen}, the right hand side of the above inequality is at most $4c_i\epsilon^2$, since the mechanism is $\epsilon$-differentially private. It remains to bound the left-hand side.



Let $S_{-i}$ denote the individuals, other than $i$,  who truthfully participate when using strategy $\sigma_\tau$. Note that by the definition of $\sigma_{\tau_{\alpha, \delta/2}}$ we have  $\E[|S_{-i}|] \geq (1-\alpha) (n-1)$, since each individual truthfully reports with probability (over his type distribution) at least $1-\alpha$.   Thus for $b_i \in \{0, 1\}$,
 $$|\E[\tp_{-i} | b_i] - \E[\pbi |b_i]| \leq \alpha.$$


Above we saw that $|\E[\tp_{-i} | b_i] - \E[\pbi |b_i]|\leq \alpha$.  Thus, by our choice of $c, d, \rho$ and Proposition~\ref{prop:pp-values},  $$\E[B_{c,d,\rho}(\tp_{-i}, \pbi)|b_i] \geq \beta = 4\epsilon^2 \tau_{\alpha, \delta/2}.$$

By assumption $\tau_{\alpha, \delta/2} \geq c_i$, and so if we can show that the expected payment that a player receives for misreporting his bit is non-positive, we will be done.

We now show that the expected payoff for lying, $ \E[B_{c,d,\rho}(\tp_{-i}, p_{1-b_i})|b_i]$ is always non-positive.   Above we saw that $|\E[\tp_{-i} | b_i] - \E[\pbi |b_i]| \leq \alpha$.  Thus, by our choice of $c, d, \rho$ and Proposition~\ref{prop:pp-values} we have that  $\E[B_{c,d,\rho}(\tp_{-i}, p_{1-b_i})|b_i] \leq 0.$ The second term is non-positive, so it follows that the payoff for lying is also non-positive.

We similarly must show that the utility associated with participation and truth-telling is greater than the utility of non-participation. That is:

$$\E[B_{c,d,\rho}(\tp_{-i}, \pbi)|b_i]- \E[v_i(M, M(\hat{b}), \hat{b}_i, \hat{b}_{-i})] \geq  - \E[v_i(M, M(\hat{b}), \bot, \hat{b}_{-i})]$$
This statement follows from analysis analogous to showing that truth-telling outperforms lying, since non-participation induces the same privacy cost as reporting $\hat{b}_i = 0$, but results in no payment.



It remains to show (2), that with probability at least $1-\delta$ the mechanism is $\alpha'$-accurate for $\alpha' = \left(\frac{\ln(\delta/2)}{\epsilon n} + \alpha\right)$-when all players play according to equilibrium $\sigma_{\tau_{\alpha, \delta/2}}$. To show this, we note that error comes from two sources: first, non-truthful behavior. However, by the definition of $\sigma_{\tau_{\alpha, \delta/2}}$, we know that with probability at least $1-\delta/2$, all but $\alpha n/2$ people truth-tell in equilibrium: that is, except with probability $1-\delta/2$, we have $|\sum_{i=1}^n\hat{b}_i - \sum_{i=1}^n b_i| \leq \frac{\alpha n}{2}$. The second source of error is the Laplace noise added to $\hat{b}$. We have that $\Pr[|\mathrm{Lap}(1/\epsilon)| \geq t/\epsilon] = \exp(-t)$. Hence, $\Pr[|\hat{b}-\bar{b}| \geq \ln(\delta/2)/\epsilon] \leq \delta/2$. Combining these two bounds gives the claim.
 \end{proof}

 Finally, we show that the total payment that the surveyor need produce tends to zero as the population size $n$ grows large. In particular, this means that the marginal cost of preserving privacy while doing peer prediction tends to zero as $n$ grows large.

 \begin{theorem}
 Fix any parameters $\alpha$ and $\delta$. In the symmetric Bayes Nash equilibrium defined by threshold strategies $\sigma_{\tau_{\alpha, \delta/2}}$, the total expected cost incurred by the surveyor:
 $$\sum_{i=1}^n \E[\pi_i] \leq \frac{4\ln(1/\delta)^2\cdot \tau_{\alpha, \delta/2}}{\alpha^2 n}\cdot \left(1 + 4\frac{\alpha}{2(\po-\pz)-4\alpha}\right) = O\left(\frac{1}{n}\right)$$
 \end{theorem}
 \begin{proof}
Recall that in equilibrium, players maximize their payoff by truthtelling. Therefore:
$$\sum_{i=1}^n \E[\pi_i] \leq \sum_{i=1}^n \E[B_{c,d,\rho}(\tp_{-i}, \pbi)|b_i]$$
However, recall that we have shown that $|\E[\tp_{-i} | b_i] - \E[\pbi |b_i]| \leq \alpha$. Therefore, we can invoke proposition \ref{prop:pp-values} (3) to conclude:
\begin{eqnarray*}
\E[B_{c,d,\rho}(\tp_{-i}, \pbi)|b_i] &\leq& \beta + 4 \rho \alpha|\pz - \po| \\
&=& \beta\cdot\left(1 + 4\frac{\alpha|\pz - \po|}{2(\po-\pz)^2-4\alpha|\pz - \po|}\right) \\
&=& \frac{4\ln(1/\delta)^2\cdot \tau_{\alpha, \delta/2}}{\alpha^2 n^2}\cdot \left(1 + 4\frac{\alpha}{2(\po-\pz)-4\alpha}\right)
\end{eqnarray*}
Therefore:
$$\sum_{i=1}^n \E[\pi_i] \leq  \frac{4\ln(1/\delta)^2\cdot \tau_{\alpha, \delta/2}}{\alpha^2 n}\cdot \left(1 + 4\frac{\alpha}{2(\po-\pz)-4\alpha}\right) = O\left(\frac{1}{n}\right). \qed$$
 \end{proof}

\section{Conclusion}
In this paper we have shown how to \emph{accurately} conduct a survey when agents have costs associated with their privacy loss that might result from the survey outcome, even when the surveyor has no ability to verify agents' private data. This result holds under even an extremely mild assumption on agents' privacy costs---that they are simply \emph{upper bounded} as a function of the privacy parameter $\epsilon$ when the mechanism satisfies $\epsilon$-differential privacy. Under a stronger, but still well-motivated assumption on privacy costs, we have further shown that remarkably, the \emph{cost} of conducting such a survey (and, in particular, the marginal cost of compensating for privacy losses) can be driven to zero simply by increasing the number of participants.

However, compared to past literature on the sensitive surveyors problem, our results come at a price. Because we cannot verify agents' bits, we inherit from the peer-prediction literature that truth-telling results in just one of possibly many Bayes Nash equilibria of our mechanism. Although the truth-telling equilibrium might be considered ``focal'', and hence likely to occur given a lack of coordination by the agents, in the presence of coordination, we might be worried that agents can collude and coordinate on a different equilibrium, which might result in lower privacy costs for them. The problem of eliminating these bad equilibria is an exciting direction for future work.

\bibliographystyle{alpha}
\bibliography{refs}

\end{document}